\documentclass[a4paper,11pt]{article}

\usepackage[T1]{fontenc}
\usepackage[utf8]{inputenc}
\usepackage{textcomp}
\usepackage{makeidx}
\usepackage{graphicx}
\usepackage{rotating}
\usepackage{float}
\usepackage{changepage}

\usepackage[toctitles]{titlesec}
\usepackage{fancyhdr}

\usepackage{hyperref}
\usepackage{amssymb}
\usepackage{amsmath,bm}
\usepackage{latexsym}
\usepackage{amsthm}
\usepackage{amssymb}
\usepackage{cases}
\usepackage{fancyvrb}

\usepackage{sectsty}
\usepackage{fancybox}
\usepackage{listings}
\usepackage{charter}

\usepackage[titletoc]{appendix}
\usepackage{geometry}

\usepackage{tikz}
\usetikzlibrary{topaths,calc}
\usepackage{enumitem}

\usepackage{epsfig}
\usepackage{tabularx}
\usepackage{latexsym}
\usepackage{alltt}
\usepackage{setspace}
\usepackage{datetime}
\newdateformat{mydate}{\monthname[\THEMONTH], \THEYEAR}

\usepackage{pslatex}
\usepackage{times}

\def\qed{\hfill$\Box$}

\newtheorem{corollary}{Corollary~}
\newtheorem{definition}{Definition~}
\newtheorem{theorem}{Theorem~}

\newtheorem{lemma}{Lemma~}
\def\abstract{{\begin{center}
\Large {\bf Abstract}
\end{center} }}

\bmdefine\bchi{\bm\chi}
\DefineVerbatimEnvironment{code}{Verbatim}{fontsize=\small}
\DefineVerbatimEnvironment{example}{Verbatim}{fontsize=\small}
\title{\bf An extremal problem in coloring of hypergraphs}
\author{Tapas Kumar Mishra \quad Sudebkumar Prasant Pal \\ Dept. of Computer Science and Engineering\\ IIT Kharagpur 721302, India}

\begin{document}
\maketitle


\begin{abstract}
Let $G(V,E)$ be a $k$-uniform hypergraph.
A hyperedge $e \in E$ is said to be properly $(r,p)$ colored by an $r$-coloring of vertices in $V$ if 
$e$ contains vertices of at least $p$ distinct colors in the $r$-coloring.
An $r$-coloring of vertices in $V$ is called a {\it strong $(r,p)$ coloring}
if every hyperedge $e \in E$ is  properly $(r,p)$ colored by the $r$-coloring.
We study the maximum number of hyperedges that can be properly $(r,p)$ colored by a single $r$-coloring
and the structures that maximizes number of  properly $(r,p)$ colored hyperedges.
\end{abstract}
\bigskip




\section{Introduction}
\label{sec:intro}
Let $G(V,E)$ be a $n$-vertex $k$-uniform hypergraph. 
A hyperedge $e \in E$ is {\it properly $(r,p)$  colored} by an $r$-coloring of vertices 
if $e$ consists of at least $p$ distinctly colored vertices.
A {\it strong $(r,p)$  coloring} of $G$ is an $r$-coloring of the vertices of $V$, such that $\forall e \in E$, $e$ consists of at least $p$ distinctly colored vertices.
We note that for a fixed $r$ and $p$, $G$ may not have any strong $(r,p)$  coloring.
Moreover, its not too hard to see that the decision problem is also $NP$-complete, since 
(i) the decision problem of bicolorability of hypergraphs is NP-complete (\cite{lovasz1973coverings}), and (ii) proper $(2,2)$ coloring of $G$
is equivalent to proper bicoloring of $G$.
Given $n$, $k$, $r$, and $p$, we study the maximum number of hyperedges $M(n,k,r,p)$ of any $n$-vertex $k$-uniform hypergraph $G(V,E)$ that can be properly $(r,p)$ colored by a single $r$-coloring.

This problem has an equivalent counterpart in graphs.
A proper coloring of a edge in graphs denotes the vertices of the edge getting different colors.
A graph is properly colored if its every edges is properly colored.
Consider an $r$-coloring of a $n$-vertex graph $H(V,E')$.
For any $K_k$ in $H$, $k \in \mathcal{N}$, a rainbow of size $x$ exists if there exists a $K_x$ which is a subgraph of   the $K_k$, $x \leq k$, and is properly colored.
Consider the problem of finding the maximum number of distinct $K_k$'s in an $r$-coloring such that each $K_k$ has a rainbow of size $p$.
It is easy to see that this problem is equivalent to the problem of finding the maximum number of hyperedges in an $n$-vertex $k$-uniform hypergraph $G(V,E)$ that can be properly $(r,p)$ colored by a single $r$-coloring: each $K_k$ is replaced by a $k$-uniform hyperedge and a rainbow of size $p$ denotes $p$ distinctly colored vertices in the hyperedge.

This problem has been motivated by the separation problems in graphs.
\begin{definition}
Let $[n]$ denote the set ${1,2,...,n}$. A set $S \subseteq [n]$ separates $i$ from $j$ if $i \in S$ and
$j \not\in S$. A set ${\cal S}$  of subsets of $[n]$ is a separator if, for each $i,j \in [n]$ with $i\neq j$,
there is a set $S$ in ${\cal S}$ that separates $i$ from $j$. If, for each $(i,j) \in  [n]*[n]$ with $i\neq j$,
there is a set $S\in{\cal S}$ that separates $i$ from $j$ and a set $T \in {\cal S}$ that separates $j$ from $i$, then ${\cal S}$ is called a complete separator. 
Moreover, with the additional constraint that the sets $S$ and $T$ that separate $i,j$ are required to be disjoint, then ${\cal S}$ is called a total separator.
\end{definition}
We refer the reader to \cite{renyi1961,kat1966,Kat1973,weg1979,dick1969,spencer1970,mao1983} for discussions and results on separating families for graphs. The notion of separation for hyperedges is introduced 
in \cite{tmspp2014}. A family ${\cal S}=\{S_1,...,S_t\}$ is called a separator for a $k$-uniform hypergraph $G(V,E)$,
$S_i \subset V$ for $1 \leq i \leq t$, such that every hyperedge $e \in E$ has a nonempty intersection with 
at least one $S_i$ and $V \setminus S_i$.
 We consider the following problem of separation for $k$-uniform hypergraphs.
Let $G(V,E)$ be a $k$-uniform hypergraph. 
 A set $S_1=\{S_{11},...S_{1r}\}$ $(r,p)$-separates a hyperedge $e \in E$ if 
 (i) $S_{1j}\subset V$, $S_{1j} \neq \phi$, $1 \leq j \leq r$, 
(ii) $\cup_{j}S_{ij}=V$, and,
 (ii)  $e$ has nonempty intersection with at least $min\{|e|,p\}$ elements of $S_1$.
 Observe that the maximum number of hyperedges that can be $(r,p)$-separated by a single family $S_1$
 is $M(n,k,r,p)$.

Consider the problem of maximizing profit between a player $P$ and an adversary $A$.
Adversary $A$ provides $n$, $k$, $r$ and $p$ to the player $P$.
$P$ performs some calculation on those parameters and finds out a number \#$e$.
Now, $A$ constructs a $n$-vertex $k$-uniform hypergraph with \#$e$ hyperedges and colors  the vertices with $r$-colors.
If $A$ can properly $(r,p)$ color at least \#$e$ hyperedges in a hypergraph, then $A$ wins.
If $A$ cannot properly $(r,p)$ color at least \#$e$ hyperedges in a hypergraph, then $P$ wins.
However, the profit of $P$ is given by $\binom{n}{k}-\text{\#}e$.
So, given a fixed $n$, $k$, $r$ and $p$, what value of \#$e$ should $P$ use so that he is guaranteed a win and his profit is maximized.
Observe that if $P$ chooses \#$e$ to be $M(n,k,r,p)+1$, then he is guaranteed a win with maximum profit.

The problem has many applications in resource allocation and scheduling.
Consider the problem where there are total $n$ resources $\{v_1,...,v_n\}$, 
$m$ processes $\{e_1,...,e_m\}$.
Each process has a distinct wish-list of  $k$ resources.
There are $r$ time slots.
A process can execute if it gets at least $p$ distinct resources in different time slots.
The problem is to maximize the number of processes that can be executed within $r$ time slots.
The solution to the above problem is equivalent to the maximum number of hyperedges that can be properly $(r,p)$ colored by a single $r$-coloring in an $n$-vertex $k$-uniform hypergraph $G(V,E)$, where $V= \{v_1,...,v_n\}$,
$E=\{e_1,...,e_m\}$.
Throughout the paper, $G$ denotes a $k$-uniform hypergraph with
vertex set $V$ and hyperedge set $E$, unless otherwise stated. 

\subsection{Motivation}
\label{subsec:motiv}
Tur\'{a}n's theorem is a fundamental result in graph theory that gives the maximum number of edges that can be present in a $K_{r+1}$ free graph. The problem was first stated by Mantel \cite{DW2001,MAGZ1998} for the special case of triangle free graphs.
He proved that the maximum number of edges in an $n$-vertex triangle-free graph is $\lfloor \frac {n^2} 4 \rfloor$.
 Tur\'{a}n \cite{MAGZ1998} posed the same problem for general $K_{t+1}$-free graphs and showed that the maximum number of 
edges in an $n$-vertex $K_{t+1}$-free graph is $(1-\frac{1}{t})\frac{n^2}{2}$.
The graph with $(1-\frac{1}{t})\frac{n^2}{2}$ edges is a Tur\'{a}n graph $T(n,t)$ - a complete $t$-partite graph with the size of partite sets differing by at most 1.
In the same spirit, Erd\'{o}s et al.\cite{MR0018807} posed the question of maximum number of edges in a graph $G(V,E)$ that does not contain some arbitrary subgraph $F$. 
\begin{definition}
Given an $k$-uniform hypergraph $F$ the Turan number $ex(n, F)$ is the maximum number of hyperedges in an $n$-vertex
$k$-uniform hypergraph not containing a copy of $F$. The Tur\'{a}n density $\pi(F)$ of F is
\begin{align*}
\pi(F)=\lim_{n \to \infty} \frac{ex(n, F)}{\binom{n}{k}}.
\end{align*}
\end{definition}
They showed that for a arbitrary graph $F$ and a fixed $\epsilon>0$, there exists a $n_0$ such that for any $n > n_0$,
\begin{align*}
(1-\frac{1}{\chi(F)-1}-\epsilon)\frac{n^2}{2}\leq ex(n,F) \leq (1-\frac{1}{\chi(F)-1}+\epsilon)\frac{n^2}{2},
\end{align*}
where $\chi(F)$ denotes the chromatic number of graph $F$.
For complete graph $K_{t+1}$, the chromatic number is $t+1$; so, the result due to Erd\'{o}s et al. for $ex(n,F)$ reduces to an approximate version of Tur\"{a}n's theorem. If $F$ is bipartite, $ex(n,H) \leq \epsilon n^2$, for $\epsilon>0$. 
This also implies that for a graph $G$, $\pi(F)=(1-\frac{1}{\chi(F)-1})$.

Having solved the problem for $F = K_t$, Tur\'{a}n \cite{tur1961} posed the natural generalization of the problem for determining $ex(n,F)$ where $F = K_t^k$ is a complete $k$-uniform hypergraph on $t$ vertices. 
The minimum number of hyperedges in an $k$-uniform hypergraph $G$ on $n$ vertices such that any
subset of $r$ vertices contains at least one hyperedge of $G$ is the 
{\it Tur\'{a}n number} $T(n, r, k)$.
Note that $G$ has this property
if and only if the {\it complementary} $k$-uniform hypergraph $G'$ is $K_r^k$-free; 
thus $T(n, r, k) + ex(n, K^k_r) = \binom{n}{k}$.
There is extensive study of both $T(n,r,k)$ and $ex(n, K^k_r)$ and we refer the reader to two surveys \cite{sidorenko95,keevash11} for details.
All the above results assumes that the host graph or hypergraph is arbitrary.
Mubayi and Talbot \cite{MubTal08}, and Talbot \cite{Talbot07} introduced Tur\'{a}n problems with coloring conditions, which could also
be viewed from the perspective of a constrained host graph.
They considered a new type of extremal hypergraph problem: given an $k$-uniform hypergraph $F$ and an integer $r \geq 2$, determine the maximum number of hyperedges in an $F$-free, $r$-colorable $r$-graph on $n$ vertices.
In similar direction, we pose the following problem:
maximize the number of hyperedges in a $r$-coloring of a $n$-vertex $k$-uniform hypergraph $G$, such that
no hyperedge of $G$ consists of less than $p$ colors.


\subsection{Our Results}
\label{subsec:results}

In order to estimate  $M(n,k,r,p)$, we first consider the case when $r$ divides $n$ and compute 
the number of distinct hyperedges that consists of exactly $p$ distinct colors 
under any balanced $r$ coloring of a $K_n^k$.
Let $m(n,k,r,p)$ denote the number of distinct hyperedges that consists of exactly $p$ distinct colors 
under any balanced $r$ coloring of a $K_n^k$. We prove the following lemma.

\begin{lemma}\label{lemma:exactmin}
For a fixed value of $n$, $k$, $r$ and $p$,
$m(n,k,r,p)= \binom{r}{p}\Big( \binom{\frac{n}{r}p}{k}- p\binom{\frac{n}{r}(p-1)}{k}+\binom{p}{2}\binom{\frac{n}{r}(p-2)}{k}
\allowbreak ... (-1)^c \binom{p}{c}\binom{\frac{n}{r}c}{k} \Big)$,
where $c$ is the smallest integer such that $\frac{n}{r}c >= k$.
\end{lemma}

Observe that summing over all the hyperedges with exactly $i$ distinct colors, $1 \leq i \leq p-1$,
we get the number of hyperedges that are colored with at most $p-1$ colors by any balanced $r$-coloring, provided $r$ divides $n$.
In Section \ref{sec:max}, we show that  the number of distinct hyperedges that consists of at least $p$ distinct colors
is maximized when the $r$-coloring is balanced. Therefore, we conclude the following theorem.

\begin{theorem}\label{thm:exactmax}
The maximum number of properly $(r,p)$ colored hyperedges of a $K_n^k$ in any $r$-coloring
(i) is $M(n,k,r,p)=\binom{n}{k} - \sum_{i=1}^{p-1} m(n,k,r,i)$,
where $m(n,k,r,i)= \binom{r}{i}\Big( \binom{\frac{n}{r}i}{k}- i\binom{\frac{n}{r}(i-1)}{k}+\binom{i}{2}\binom{\frac{n}{r}(i-2)}{k}... (-1)^c \binom{i}{c}\binom{\frac{n}{r}c}{k} \Big)$, and,
$c$ is the smallest integer such that $\frac{n}{r}c >= k$, and,
(ii) the $r$-coloring that maximizes the number of properly colored hyperedges splits the vertex set into equal sized parts,
 provided $r$ divides $n$.
\end{theorem}

Furthermore, we generalize the above theorem for arbitrary $n$ i.e. to cases where $n$ does not divide $r$ and derive a upper and lower bound for $M(n,k,r,p)$ as given by the following theorem.

\begin{theorem}\label{thm:boundminmax}
For a fixed $n$, $k$, $r$ and $p$, the maximum number of of properly $(r,p)$ colored $k$-uniform hyperedges $M(n,k,r,p)$ on any $n$-vertex hypergraph $G$ is at most  $\binom{n}{k} - \sum_{i=1}^{p-1} m(n_1,k,r,i)$ and at least 
$\binom{n}{k} - \sum_{i=1}^{p-1} m(n_2,k,r,i)$,
where  $n_1=\lfloor\frac{n}{r} \rfloor \cdot r$, $n_2=\lceil \frac{n}{r} \rceil \cdot r$, and
$m(n',k,r,i)= \binom{r}{i}\Big( \binom{\frac{n'}{r} i}{k}- i\binom{ \frac{n'}{r} (i-1)}{k}+\binom{i}{2}\binom{ \frac{n'}{r} (i-2)}{k}... (-1)^c \binom{i}{c}\binom{ \frac{n'}{r}  c}{k} \Big)$, and,
$c$ is the smallest integer such that $\frac{n'}{r} c >= k$.
Moreover,  the number of properly $(r,p)$ colored hyperedges is maximized when the $r$-coloring is balanced.
\end{theorem}

\subsection{Notations}
\label{subsec:notations}

1. For a set $A$,  $\binom{[A]}{r}$ denotes the set of all the distinct $r$-element subsets of $A$. For instance, $\binom{[n]}{r}$ denotes the set of all the distinct $r$-element subsets of $\{1,...,n\}$,
$|\binom{[n]}{r}| = \binom{n}{r}$.

\noindent 2. For a set $S=\{S_1,...,S_l\}$, for any fixed $l$, $U(S)$ denotes the union of the elements, 
i.e $U(S)=S_1 \cup ...\cup S_l$.

\noindent 3. {\bf Lexicographic ordering}. Consider a $n$-element set $V=\{v_1,...,v_n\}$ and a set of $k$-element subsets $E=\{e_1,...,e_m\}$ of $V$, where 
$e_i \subset V$, for $1 \leq i \leq m$. 
For any $v_q,v_r \in V$, $v_q \prec v_r$ if $q \leq r$. 
Let $e_i,e_j \in E$, where $e_i =\{v_{i1},...,v_{ik}\}$ and $e_j =\{v_{j1},...,v_{jk}\}$. Then, $e_i \prec e_j$ if
there exists an index $l$ such that $v_{i1}=v_{j1}$,...,$v_{i(l-1)}=v_{j(l-1)}$ and $v_{il} \prec v_{jl}$
An ordering $O$ of subsets of $E$ is a lexicographic ordering if for every $e_i,e_j \in O$, $e_i$ precedes $e_j$ in $O$ if and only if
   $e_i \prec e_j$.
\section{Exact Number of properly $(r,p)$ colored hyperedges in a balanced partition}
\label{sec:exact}
Let $G(V,E)$ be a $n$-vertex $k$-uniform hypergraph, where $V$ denotes the vertex set and $E$ denotes the set of hyperedges.
An $r$-coloring $X$ of vertices in $V$ partitions the vertex set into $r$ color classes $A=\{A_1,...,A_r\}$,
where $A_j \subseteq V$, $1 \leq j \leq r$ and every vertex $v \in A_j$ receives the same color under $X$.
An $r$-coloring of vertices is called balanced if every color class is of almost same size,
i.e. for all $A_j \in A$, $|A_j|= \lceil \frac{n}{r}\rceil$ or  $|A_j|= \lfloor \frac{n}{r}\rfloor$.
Let $p$ be some fixed integer, $1<p \leq r$ and $p \leq k$.
In this section, we study the number of distinct hyperedges that consists of exactly $p$ distinct colors 
under any balanced $r$ coloring of $G$.
Throughout the section, we assume that $n$ is divisible by $r$, such that 
for all $A_j \in A$, $|A_j|=\frac{n}{r}$. 

Consider a balanced $r$ coloring $X$ of vertices a 
$K_n^k$.
Let $A=\{A_1,...,A_r\}$ denote the corresponding color partition.
Let $m(n,k,r,p)$ denote the number of distinct hyperedges that consists of exactly $p$ distinct colors 
under $X$.
Let $B$ be the set of all the  $p$-element subsets $B_i$ of $A$, $1 \leq i \leq \binom{r}{p}$ i.e.
$B=\{B_i|B_i \text{ is the $i$th $p$-element subset of }\binom{[A]}{p} \}$.
Consider the $i$th $p$-element subset $B_i \in B$.
Let $m_i(n,k,r,p)$ denote the number of distinct hyperedges $e \in E$ 
that consists of exactly $p$ distinct colors 
under $X$ and $e \subseteq U(B_i)$.
Due to the balanced nature of the $r$-coloring $X$, 
note that $m_i(n,k,r,p)=m_l(n,k,r,p)$, for any $B_i,B_l \in B$.
Observe that 
\begin{align}\label{eq:total}
m(n,k,r,p) = \sum_{B_i \in B} m_i(n,k,r,p)= \binom{r}{p} m_i(n,k,r,p).
\end{align}
So, we focus our attention on computing $m_i(n,k,r,p)$ for a fixed $p$-element subset $B_i \in B$.
Without loss of generality, we consider $B_1=\{A_1,...,A_p\}$ as the fixed $p$-element subset of  $B$
and compute $m_1(n,k,r,p)$.

There are exactly $p$ subsets of size $p-1$ of $B_1$.
Let these sets be $P_1,...,P_p$, in the lexicographic order.
Let $N(P_j)$ denote the number of hyperedges $e \in E$ such that $e \subseteq U(P_j)$, $P_j \in B_1$,
and let $N(P_{j}...P_{l})$ denote the number of hyperedges $e \in E$ such that $e \subseteq U(P_{j}) \cap ... \cap 
U(P_{l})$ and, $P_{j},...P_{l} \in B_1$.	
Observe that $N(P_j)= \binom{\frac{n}{r}(p-1)}{k}$, $1 \leq j \leq p$.
So, 
\begin{align}\label{ineq:p-1}
\sum_{1\leq j \leq p} N(P_j) = p\binom{\frac{n}{r}(p-1)}{k}.
\end{align}

Note that if $e \subseteq U(P_{j1})$ and $e  \subseteq U(P_{j2})$, then $e \subseteq U(P_{j1}) \cap U(P_{j2})$. 
Observe that $P_{j1}$ and $P_{j2}$ can have at most $p-2$ parts in common; $e \subseteq U(P_{j1} \cap P_{j2})$
implies that $e$ lies in a fixed subset of $p-2$ parts of $P_{j1}$, that is also a subset of $P_{j2}$.
So, number of hyperedges $e$ that lie in a fixed $p-2$ parts $P_{j1} \cap P_{j2}$ is
 $N(P_{j1}P_{j2}) = \binom{\frac{n}{r}(p-2)}{k}$.
Since there are exactly $\binom{p}{2}$ distinct pairs of the form $\{P_{j1},P_{j2}\}$,
total number of hyperedges $e$ that are subsets of $p-2$-sized subsets of $B_1$ is
\begin{align}\label{ineq:p-2}
\sum_{1\leq j1 < j2 \leq p} N(P_{j1} P_{j2}) = \binom{p}{2}\binom{\frac{n}{r}(p-2)}{k}.
\end{align}

Let $c$ be the smallest integer such that $\frac{n}{r}c >= k$. Then, $\frac{n}{r}(c-1) < k$.
Consider any fixed $c-1$ parts $A_{j1},...,A_{j(c-1)}$. Observe that for any hyperedge $e \in E$, $e \not\subseteq 
A_{j1} \cup .... \cup A_{j(c-1)}$.
So, we compute all the summations of the form $\sum_{1\leq j1 < j2...< jc \leq p} N(P_{j1} P_{j2}...P_{jc})$
till $\frac{n}{r}c \geq k$, where 
\begin{align}\label{ineq:c}
\sum_{1\leq j1 < j2...< jc \leq p} N(P_{j1} P_{j2}...P_{jc}) = \binom{p}{c}\binom{\frac{n}{r}c}{k}.
\end{align}

Observe that if a hyperedge $e$ is a subset of $B_1$ and is not a subset of any of the $P_j$,
 $P_j \in \{P_1,...,P_p\}$, then $e$ consists of exactly $p$ colors in the $r$-coloring.
The total number of hyperedges $e \subseteq B_1$ is $N(B1)=\binom{\frac{n}{r}p}{k}$.
So, by definition, $N(P'_1,...,P'_p)$ denotes all the hyperedges $e \subseteq B_1$ such that 
$e$ consist of exactly $p$ colors, i.e. $m_1(n,k,r,p)=N(P'_1,...,P'_p)$.
In order to compute $m_1(n,k,r,p)$, we use the fundamental result of inclusion exclusion stated below.
\begin{theorem}\cite{rosen2002}
\label{thm:incexc}
Let $A$ be any $n$-element set, and let $P_1$, ..., $P_m$ denote $m$ properties of elements of $A$.
Let $A_i \subset A$ is the subset of 
elements of $A$ with property $P_i$.
Let $N(P_i)$ denote the number of elements of $A$ with property $P_i$, i.e. $N(P_i)=|A_i|$, for $1 \leq i \leq m$.
Let $N(P_iP_j...P_l)=|A_i \cap A_j \cap ...\cap A_l|$.
Let $N(P_i')$ denote the number of elements of $A$ that does not satisfy property $P_i$ and 
the number of elements with none of the properties $P_i$,$P_j$...,$P_l$ is denoted by
$N(P_i'P_j'...P_l')$.
Then,
\begin{align}
N(P_1'P_2'...P_m')=n- \sum_{1\leq i \leq m} N(P_i)+\sum_{1\leq i < j \leq m} N(P_iP_j)-...+ 
(-1)^{m} N(P_1P_2...P_m).
\end{align}
\end{theorem}

So, using principle of inclusion exclusion \ref{thm:incexc}, we have,
\begin{align}
N(P'_1,...,P'_p)= &N(B1) - \sum_{1\leq j \leq p} N(P_j)+\sum_{1\leq j1 < j2 \leq p} N(P_{j1} P_{j2})- ... (-1)^c \sum_{1\leq j1 < j2...< jc \leq p} N(P_{j1} P_{j2}...P_{jc}) \nonumber\\
=& \binom{\frac{n}{r}p}{k}- p\binom{\frac{n}{r}(p-1)}{k}+\binom{p}{2}\binom{\frac{n}{r}(p-2)}{k}... (-1)^c
	\binom{p}{c}\binom{\frac{n}{r}c}{k}.
\end{align}

Now, using Equation \ref{eq:total}, we get,
$m(n,k,r,p)= \binom{r}{p}\Big( \binom{\frac{n}{r}p}{k}- p\binom{\frac{n}{r}(p-1)}{k}+\binom{p}{2}\binom{\frac{n}{r}(p-2)}{k}... (-1)^c \binom{p}{c}\binom{\frac{n}{r}c}{k} \Big)$,
where $c$ is the smallest integer such that $\frac{n}{r}c >= k$.
This concludes the proof of Lemma \ref{lemma:exactmin}.

Observe that summing over all the hyperedges with exactly $i$ distinct colors, $1 \leq i \leq p-1$,
we get the number of hyperedges that are colored with at most $p-1$ colors by any balanced $r$-coloring, provided $r$ divides $n$. Therefore, the exact number of properly $(r,p)$ colored hyperedges in a balanced partition is
\begin{align}\label{eq:M1}
M(n,k,r,p)=\binom{n}{k} - \sum_{i=1}^{p-1} m(n,k,r,i).
\end{align}
Consider the case when $r=p=2$, i.e., when we are performing a bicoloring on $n$ vertices and proper coloring of a hyperedge $e$ denote $e$ becoming non-monochromatic under the bicoloring.
Observe that $M(n,k,2,2)=\binom{n}{k} - m(n,k,2,1)$, and $m(n,k,2,1) = 2*\binom{\frac{n}{2}}{k}$.
Therefore, $M(n,k,2,2)=\binom{n}{k} - 2*\binom{\frac{n}{2}}{k}$, which agrees with the existing results.
Note that $M(n,k,r,p)$ is a non-decreasing function of $n$. So, $M(n-1,k,r,p)\leq M(n,k,r,p) \leq M(n+1,k,r,p)$.

Let $x(i,j,n,k,r)=\binom{r}{i}\binom{\frac{n}{r}i}{k} - \frac{r-j}{i-j+1}x(i,j-1,n,k,r)$.
 $x(i,j,n,k,r)$ denotes the number of hyperedges that are colored with less than or equal to $j$ colors by an $r$-coloring,
when counted with respect to color classes of size $i$, $i\geq j$.
Here, the term $\binom{r}{i}\binom{\frac{n}{r}i}{k}$ accounts for every hyperedge $e \in E$, that is a subset of
some fixed $i$ color parts of the $r$-coloring.
Any $(j-1)$-sized color parts are repeated $r-j+1$ times when counted over all $j$-sized color classes; however, we need to count it exactly once.
Each hyperedge inside some fixed $i$-sized set is counted $i-j+1$ times over all the $j-1$ sized sets.
So, $\binom{r}{i}\binom{\frac{n}{r}i}{k} - \frac{r-j+1}{i-j+1}x(i,j-1,n,k,r)+\frac{1}{i-j+1}x(i,j-1,n,k,r)$ counts 
the number of hyperedges that are colored with less than or equal to $j$ colors by an $r$-coloring,
when counted with respect to color classes of size $i$, $i\geq j$. $\frac{1}{i-j+1}x(i,j-1,n,k,r)$ term is added in order to include the hyperedges colored with less than or equal to $j-1$ colors.
Observe that $x(p-1,p-1,n,k,r)$ denotes the number of hyperedges colored with less than or equal to $p-1$ colors by a balanced $r$-coloring.
Therefore, 
\begin{align}
M(n,k,r,p)=\binom{n}{k} - x(p-1,p-1,n,k,r).
\end{align}

\section{Maximizing the number of properly $(r,p)$ colored hyperedges}
\label{sec:max}

In this section, we show that the number of properly $(r,p)$ colored hyperedges is maximized when the $r$-coloring is balanced. We show that the number of hyperedges colored with less than or equal to $p-1$ colors is minimized for a balanced $r$-coloring, thereby proving the above claim.

Consider an $r$-coloring $X$ of vertices a 
$K_n^k$.
Let $A=\{A_1,...,A_r\}$ denote the corresponding color partition and let $|A_i|=n_i$, for $1 \leq i \leq r$.
Let $m_{|X}(n,k,r,p)$ denote the number of distinct hyperedges that consists of at most $p$ distinct colors 
under $X$.
Let $n_1 \geq n_2+2$. Then we have the following lemma.
\begin{lemma}\label{lemma:min}
The number of hyperedges colored with at most $p$  colors is reduced by moving a vertex $v \in A_1$ from $A_1$ to $A_2$, i.e. switching the color of $v$ from 1 to 2 produces an $r$-coloring $X'$ such that 
$m_{|X'}(n,k,r,p) < m_{|X}(n,k,r,p)$.
\end{lemma}

\begin{proof}
In order to prove that $m_{|X'}(n,k,r,p) < m_{|X}(n,k,r,p)$, we analyze: 
(i) the {\it gain } $g$: the number of hyperedges $e \in E$ such that $e$ is colored with greater than $p$ colors under $X$ and $e$ receives at most $p$ colors under $X'$, and,
(ii) the {\it loss} $l$: the number of hyperedges $e \in E$ such that $e$ is colored with at most $p$ colors under $X$ and $e$ receives at least $p+1$ colors under $X'$.
Note that a hyperedge $e\in E$ contributes to $g$ or $l$ if and only if 
$v \in e$.
Since $m_{|X'}(n,k,r,p) = m_{|X}(n,k,r,p)+g-l$,
in order to prove Lemma \ref{lemma:min}, we need to show that $l>g$.

Let $y(n,k,r,p)$ denote the minimum number of $k$-uniform hyperedges on $n$ labeled vertices that are colored with exactly $p$ colors by any $r$ coloring.
Observe that a hyperedge $e\in E$ contributes to $g$ if and only if 
it consists of exactly $p+1$ colors in $X$, 
$v \in e$ and includes no other vertex from  $A_1$, i.e., $e \cap A_1 = v$, and
includes at least one vertex from $A_2$, i.e., $e \cap A_2 \geq 1$.
So, gain due to switching $v$ from $A_1$ to $A_2$ is
\begin{align}
g = \sum_{i=1}^c \binom{n_2}{i} y(n-n_1-n_2,k-i-1,r-2,p-1),
\end{align}
where $c$ be the smallest integer such that $\frac{n}{r}c >= k$.
In each of the $c$ terms in the summation, $\binom{n_2}{i}$ denotes the number of ways to choose exactly $i$ vertices from $A_2$ (of color 2), $y(n-n_1-n_2,k-i-1,r-2,p-1)$ denotes the minimum number of hyperedges that can be formed 
consisting of  exactly $k-(i+1)$ vertices from $A \setminus (A_1 \cup A_2)$ and exactly $p-1$ distinct colors.
The $k-(i+1)$ vertices from $A \setminus (A_1 \cup A_2)$ with $p-1$ distinct colors combined with $i$ vertices from $A_2$ and $v$ from $A_1$ forms the hyperedges $e$ consisting of exactly $p+1$ colors under coloring $X$ 
including $v$, $e \cap A_1 = v$,
and $|e \cap A_2|=i$.

Similarly,  a hyperedge $e \in E$ contributes to $l$ if and only if 
it consists of exactly $p$ colors in $X$, 
includes no other vertex from  $A_2$, i.e., $e \cap A_2 = \phi$, and
$v \in e$ and includes at least one vertex other than $v$ from $A_1$, i.e., $|e \cap A_1| \geq  2$.
So, loss due to switching $v$ from $A_1$ to $A_2$ is
\begin{align}
l = \sum_{i=1}^c \binom{n_1 - 1 }{i} y(n-n_1-n_2,k-i-1,r-2,p-1).
\end{align}
Since $n_1 \geq n_2+2$, $n_1-1 > n_2$. So, comparing $l$ and $g$ term-wise, we get $l > g$ as desired.
\qed
\end{proof}

Lemma \ref{lemma:min} implies that the number of hyperedges colored with less than $p$ colors can be minimized until
the color partition $\{A_1,...,A_r\}$ is balanced, i.e. for every $i$, $1 \leq i \leq r$, 
$\lfloor \frac{n}{r} \rfloor \leq |A_i| \leq \lceil\frac{n}{r}\rceil$.
Therefore, the number of properly $(r,p)$ colored hyperedges is maximized when the $r$-coloring is balanced.
So, using Equation \ref{eq:M1}, Theorem \ref{thm:exactmax} follows.

Observe that even if $r$ does not divide $n$, the $r$-coloring that maximizes the number of properly colored hyperedges splits the vertex set into almost equal sized parts (from Lemma \ref{lemma:min}) of either $\lfloor\frac{n}{r} \rfloor $ or $\lceil \frac{n}{r} \rceil$ size. 
Therefore, we can get a upper bound on $M(n,k,r,p)$
by computing the minimum number of hyperedges including vertices of at most $p-1$ distinct colors 
with $\lfloor\frac{n}{r} \rfloor \cdot r$ vertices and subtracting from $\binom{n}{k}$.
Furthermore,  we can get a lower bound on $M(n,k,r,p)$
by computing the minimum number of hyperedges including vertices of at most $p-1$ distinct colors 
with $\lceil \frac{n}{r} \rceil \cdot r$ vertices and subtracting from $\binom{n}{k}$.
This observation combined with Theorem \ref{thm:exactmax} proves Theorem \ref{thm:boundminmax}.

For the special case when $r = p = k$, we can compute $M(n,k,r,p)$ much easily.
Observe that any hyperedge must contain one vertex each from each of the color classes $\{A_1,...A_r\}$
in order to be properly $(r,p)$ colored. So, the number of properly colored hyperedges under any $r$-coloring
is $|A_1||A_2|...|A_r|$. Using the second part of Theorem \ref{thm:boundminmax},
$M(n,k,r,p)=|A_1||A_2|...|A_r|$, where $\{A_1,...A_r\}$ is a balanced partition.
So, we have the following corollary.
\begin{corollary}\label{cor:exactmax}
The number of properly $(r,p)$ colored hyperedges of a $K_n^k$ in any $r$-coloring
 is $|A_1||A_2|...|A_r|$ when $r=p=k$. Moreover,
the $r$-coloring that maximizes the number of properly colored hyperedges splits the vertex set into almost 
equal sized parts.
\end{corollary}

\bibliographystyle{plain}
\bibliography{MAX_references}

\end{document}